\documentclass[11pt]{article}
\usepackage{amsmath, amssymb, amsfonts, amsthm}

\newtheorem{theorem}{Theorem}[section]
\newtheorem{proposition}[theorem]{Proposition}

\newtheorem{lemma}[theorem]{Lemma}

\newcommand{\be}{\begin{equation}}
\newcommand{\ee}{\end{equation}}
\newcommand{\bey}{\begin{eqnarray}}
\newcommand{\eey}{\end{eqnarray}}

\newcommand{\ph}{\varphi}

\newcommand{\bR}{{\mathbb R}}

\newcommand{\tr}{\mbox{Tr}}

\newcommand{\wt}{\widetilde}

\newcommand{\cE}{{\cal E}}

\input epsf

\newcommand{\donothing}[1]{}

\begin{document}

\title{Dynamics of Bose-Einstein condensates of fermion pairs in the low density limit of BCS theory}

\author{Christian Hainzl and Benjamin Schlein\thanks{Partially supported by an ERC Starting Grant} 
\\ \\  Mathematical Institute, University of T{\"u}bingen \\ Auf der Morgenstelle 10, 72076 T\"ubingen, Germany \\ \\ 
Institute of Applied Mathematics, University of Bonn\\ Endenicher Allee 60, 53115 Bonn, Germany}

\maketitle

\begin{abstract}
We show that the time-evolution of the wave function describing the macroscopic variations of the pair density in BCS theory can be approximated, in the dilute limit, by a time-dependent Gross-Pitaevskii equation. 
\end{abstract}

\section{Introduction and main result}
\label{sec:intro}
\setcounter{equation}{0}

We consider systems of fermions interacting through a two-body potential admitting a bound state with negative energy. At very low temperature and density, the fermions are then expected to form tightly bounded pairs, which behave like bosons and produce a Bose-Einstein condensate. {F}rom the mathematical point of view, it would be very interesting to establish the validity of this picture starting from first principle many body quantum mechanics. This seems however a very difficult task. Recently, this problem has been studied in \cite{HS}, starting from the Bardeen-Cooper-Schrieffer (BCS) approximation of many body quantum mechanics. In this paper, the authors show that at zero temperature the macroscopic variations in the pair density are described (to leading order) by the 
Gross-Pitaevskii energy functional, corresponding to a Bose-Einstein condensate of fermion pairs with an effective repulsive interaction. In the present paper, we analyze the same regime, but from a dynamical point of view. We consider the time evolution, governed by the time-dependent BCS equation, of an initial BCS state with sufficiently small energy. The energy condition guarantees that fermions form pairs, and it fixes the microscopic structure of the pair density (see Proposition \ref{prop} below). We assume that, at time $t=0$, the macroscopic variations of the pair density are described by a wave function $\psi \in L^2 (\bR^3)$. Then we show that, at time $t \not =0$, on the macroscopic scale, the pair density is described by a new wave function $\psi_t \in L^2 (\bR^3)$, given, to leading order, by the solution of the time-dependent Gross-Pitaevskii equation with initial data $\psi_{t=0} = \psi$.  

\medskip

As pointed out above, our analysis is relevant for dilute systems at or close to zero temperature. If instead one considers the regime close to the critical temperature, the macroscopic variations of the pair density are described by the Ginzburg-Landau energy functional; this has been recently proven, starting again from the BCS approximation, in \cite{FHSS}, making use of earlier works \cite{HHSS,FHNS,HS2} where, among others, the existence of a unique critical temperature was established for systems interacting via a general class of two body interactions. 
Let us now summarize the physical picture. In the seminal paper \cite{Leg} Leggett suggested that the BCS functional, introduced by Bardeen, Cooper and Schrieffer in \cite{BCS}, serves as a valuable description of a large class of Fermi systems, with interactions ranging from weak two-particle potentials up to rather strong interactions allowing for bound states. For weak attractions, fermions tend to form Cooper pairs, living on a much larger scale compared with the mean particle distance, and displaying superfluidity (superconductivity). In the case of strong potentials, the fermions are forced into tightly bound pairs behaving as bosons and forming a Bose-Einstein condensate. This picture was later extended to positive temperature in \cite{NRS}. Whereas the BCS functional is supposed to remain valid in the whole crossover region from weak to strong coupling, two distinct parameter regimes are of particular interest, as pointed out in \cite{melo,DW}. One is the low density and low temperature limit which, for strong coupling, leads to the Gross-Pitaevskii theory (the BEC regime). The other one is the limit close to the critical temperature, where the macroscopic variations are captured by the Ginzburg-Landau theory. For nice reviews on the physical background of this subject we refer to \cite{randeria, zwerger}. Whereas the emergence of these effective models on the macroscopic scale was proven in \cite{HHSS,HS} in the static case,  the goal of the present paper is to settle the evolution problem on the BEC side.  A very interesting open problem, which is not addressed in the present paper, consists in understanding (starting from the BCS approximation or, even more ambitiously, from many-body quantum mechanics) the emergence of a time-dependent Ginzburg-Landau type equation for the description of the macroscopic evolution in the regime close to the critical temperature. 

\medskip

In BCS-theory, the state of the fermionic system is described by a self-adjoint operator $\Gamma : L^2 (\bR^3) \oplus L^2 (\bR^3) \to L^2 (\bR^3) \oplus L^2 (\bR^3)$, satisfying $0 \leq \Gamma \leq 1$, defined by the $2 \times 2$ matrix
\begin{equation}\label{eq:Gamma} \Gamma = \left( \begin{array}{ll} \gamma & \alpha \\ \overline{\alpha} & 1- \overline{\gamma} \end{array} \right)\, , \end{equation}
where $\gamma,\alpha : L^2 (\bR^3) \to L^2 (\bR^3)$ will be described in terms of their integral kernels $\gamma (x,y)$, $\alpha (x,y)$. The bar indicates complex conjugation, i.e. $\overline{\alpha} (x,y) = \overline{\alpha (x,y)}$. The fact that $\Gamma$ is hermitian implies that $\gamma$ is hermitian and that $\alpha$ is symmetric, i.e. $\gamma (x,y) = \overline{\gamma(y,x)}$ and $\alpha (x,y) = \alpha (y,x)$. Moreover, the condition $0 \leq \Gamma \leq 1$ implies $\Gamma^2 \leq \Gamma$ and thus $0 \leq \gamma \leq 1$ and $\alpha \overline{\alpha} \leq \gamma (1- \gamma)$. There are
no spin variables in $\Gamma$. We implicitly use $SU(2)$ invariance of the states. In fact, one has to imagine that the full, spin dependent, Cooper pair wave
function is the product of $\alpha(x,y)$ with an antisymmetric spin singlet.

\medskip


We are going to assume particles interact through a two-body potential $V$ admitting a negative energy bound state. More precisely, we will assume that $V$ satisfies the following assumptions.

\bigskip

\noindent {\bf Condition 1.} We assume $V \in L^1 (\bR^3) \cap L^\infty (\bR^3)$, with $V(-x) = V(x)$. Moreover, we assume that $-2\Delta + V$ has an isolated ground state energy $-E_b < 0$ with a  unique non-negative normalized ground state $\alpha_0$, satisfying $(-2\Delta + V) \alpha_0 = - E_b \alpha_0$. Under these assumptions, $\alpha_0 (-x) = \alpha_0 (x)$, $\alpha_0 \in L^1 (\bR^3) \cap L^\infty (\bR^3)$ and $|x|^6 \alpha_0 (x) \in L^\infty (\bR^3)$. 

\bigskip

\noindent We are interested in the ultradilute regime, where $N \gg 1$ fermions move in a volume of order $N^3$ (we consider a system defined on $\bR^3$; the length-scale is determined here by the choice of the initial data). The particles are also subjected to an external potential, which we assume to be weak and slowly varying (it varies only over macroscopic scales of order $N$). It is useful to rescale lengths, introducing macroscopic coordinates. Putting $h = 1/N \ll 1$ (so that the expected number of particles in the system is $1/h$), the BCS energy functional  in macroscopic units is given by
\[ \cE_{\text{BCS}} (\Gamma) = \tr \, \left( -h^2 \Delta + h^2 W \right) \gamma + \frac{1}{2}  \int dx dy \, V\left( \frac{x-y}{h} \right) |\alpha (x,y)|^2. \]
For a formal derivation of the BCS functional from quantum mechanics, see e. g. \cite[Appendix A]{HHSS}.  In \cite{HS}, in a slightly different setting (with periodic boundary conditions), it was proven that, for $h \ll 1$,  
\[ \inf_{0 \leq \Gamma \leq 1 ; \tr \, \gamma = 1/h} \cE_{\text{BCS}} (\Gamma) = - \frac{E_b}{2h} + h \inf_{\|\psi \|_2 = 1} \cE_{\text{GP}} (\psi) + o(h) \]
with the Gross-Pitaevskii energy functional
\[ \cE_{\text{GP}} (\psi) = \int  dx \, \left(\frac{1}{2} |\nabla \psi (x)|^2 + 2 W(x) |\psi (x)|^2 + g |\psi (x)|^4 \right)\]
and the coupling constant
\begin{equation}\label{eq:g} 
g = \frac{1}{(2\pi)^3} \int dq \, |\widehat{\alpha}_0 (q)|^4 \, (2q^2 + E_b).
\end{equation}
It was also shown in \cite{HS} that the pair density of every approximate minimizer of the BCS functional has the form 
\begin{equation}\label{eq:alpha-hs} \alpha (x,y) = \frac{1}{h^2} \psi \left(\frac{x+y}{2} \right) \alpha_0 \left( \frac{x-y}{h} \right) + \xi (x,y) \end{equation} where $\psi$ is an approximate minimizer of the GP functional, and where $\| \xi \|_{2} \leq C h^{1/2}$ and $\| \xi \|_{H^1} \leq C h^{-1/2}$ as $h \to 0$ (for comparison, $\| \alpha \|_2 \simeq h^{-1/2}$ and $\| \alpha \|_{H^1} \simeq h^{-3/2}$, as $h \to 0$).  The interpretation of the results of \cite{HS} is straightforward: at zero temperatures, particles will minimize their energy by forming pairs, each having energy $-E_b$ (the main term in the energy, $-E_b/2h$, is exactly the energy of the $1/2h = N/2$ pairs). The pairs will then behave like bosons and, in the ultradilute limit, they will form a Bose-Einstein condensate, with condensate wave function given by the minimizer of the GP-functional (the second term in the energy, given by $h  \inf_{\|\psi \|_2 = 1} \cE_{\text{GP}} (\psi)$, is the energy of the condensate). 

\medskip

In the present manuscript we consider the dynamics in the ultradilute limit, as determined by the BCS theory. We are going to study the evolution of BCS states with energy bounded above by
\begin{equation}\label{eq:bd-en} \cE_{\text{BCS}} (\Gamma) \leq -\frac{E_b}{2h} + C h \end{equation}
in the limit of small $h \ll 1$. This assumption on the energy implies that the fermions will still form pairs or, equivalently, that the microscopic structure of $\alpha (x,y)$ is still given by $\alpha_0 ((x-y)/h)$, as in (\ref{eq:alpha-hs}). In other words, (\ref{eq:bd-en}) guarantees that only macroscopic excitations of the pair density are allowed, because microscopic excitations would cost too much energy. 
Since the relative energy $\Delta E/E$ of the excitations we are considering is of the order $h^2$, we have to consider times of the order $1/h^2$ in order to observe a non-trivial dynamics. With respect to rescaled time, the (Hamiltonian) time evolution generated by the BCS functional $\cE_{\text{BCS}}$ is governed by the equation
\begin{equation}\label{eq:BCS-dyn} i h^2 \partial_t \Gamma_t  =  \left[ H_{\Gamma_t} , \Gamma_t \right] \end{equation}
for the time-dependent operator $\Gamma_t: L^2 (\bR^3) \oplus L^2 (\bR^3) \to L^2 (\bR^3) \oplus L^2 (\bR^3)$. The self-consistent Hamiltonian $H_{\Gamma_t}$ is given here by
\[ H_{\Gamma_t} = \left( \begin{array}{ll} -h^2 \Delta + h^2 W  & V\alpha_t \\ V \bar{\alpha}_t & h^2 \Delta - h^2 W \end{array} \right). \]
A similar equation emerging in the study of the dynamics of stars was recently studied in \cite{HLLS}. As discussed in \cite{HLLS} the proof of existence of local in time solutions of \eqref{eq:BCS-dyn} is standard. Since local solutions preserve the energy, and since the potential is assumed to be bounded, one immediately obtains global well-posedness for (\ref{eq:BCS-dyn}), for initial data in the energy space. Eq. (\ref{eq:BCS-dyn}) translates into the following two coupled nonlinear equations for the kernels $\gamma_t (x,y)$ and $\alpha_t (x,y)$:
\begin{equation}\label{eq:BCS-dyn2} \begin{split}
i h^2 \partial_t \, \gamma_t = & \; [-h^2\Delta + h^2 W, \gamma_t] + i G_{\alpha_t} \\
i h^2 \partial_t \, \alpha_t (x,y) = & \; \left( -h^2 \Delta_x +h^2 W(x) - h^2 \Delta_y +h^2 W(y) + V \left( \frac{x-y}{h} \right) \right) \alpha_t (x,y) \\ &- \int dz \gamma_t (x,z) V\left( \frac{z-y}{h} \right) \alpha_t (z,y) - \int dz \gamma_t (y,z) V \left( \frac{z-x}{h} \right) \alpha_t (x,z)
\end{split} \end{equation}
with the operator $G_{\alpha_t}$ defined by its integral kernel
\[ G_{\alpha_t} (x,y) = i \int dz \alpha (x,z) \overline{\alpha} (y,z) \left( V\left(\frac{y-z}h\right) - V\left(\frac{x-z}h\right) \right) \]
Note that $\tr \, \gamma_t$ is preserved, because $\tr \, G_{\alpha_t} = 0$ for all $t \in \bR$. 
In the following we will focus on the equation for $\alpha_t$. It is useful to represent the kernel $\alpha_t$ using center of mass and relative coordinates. For this reason, we define
\[ \wt{\alpha}_t (X,r) := \alpha_t \left(X + r/2 , X - r/2 \right)  \] 
which means that $ \alpha_t (x,y) = \wt{\alpha}_t ((x+y)/2 , x-y)$. {F}rom (\ref{eq:BCS-dyn}), we obtain then the following equation for the evolution of $\wt{\alpha}_t$: 
\begin{equation}\label{eq:ev-wtalpha}
\begin{split}
ih^2 \partial_t\wt{\alpha}_t (X,r) = \; & \left( -\frac{h^2}{2}  \Delta_X - 2 h^2 \Delta_r + V(r/h) \right)  \wt{\alpha}_t (X,r) \\ &+ h^2 \left( W(X+r/2) + W(X-r/2) \right) \wt{\alpha}_t (X,r)  \\ &- \int dz \, \gamma_t (X+r/2, X+z -r/2) V\left( z/h \right) \wt{\alpha}_t (X-r/2 + z/2,z) \\ &- \int dz \, \gamma_t (X -r/2, X-z+r/2) V \left(z/h \right) \wt{\alpha}_t (X+r/2-z/2,z)
\end{split} \end{equation}
We rewrite the equation (\ref{eq:ev-wtalpha}) in integral (or Duhamel) form:
\begin{equation}\label{eq:wtalpha-duh}
\begin{split}
\wt{\alpha}_t (X,r) = \, &U (t) \wt{\alpha}_0 (X,r) + \frac{1}{i} \int_0^t ds \, U (t-s) \left( W(X+r/2) + W (X-r/2) \right) \wt{\alpha}_s (X,r)  \\ &- \frac{1}{ih^2} \int_0^t ds \, U (t-s) \int dz \, \gamma_s (X+r/2, X+z -r/2) V(z/h) \wt{\alpha}_s (X-r/2 +z/2, z) \\ &- \frac{1}{ih^2} \int_0^t ds \, U (t-s) \int dz \, \gamma_s (X -r/2, X-z+r/2)  V(z/h) \wt{\alpha}_s (X-r/2 -z/2, z)
\end{split} 
\end{equation}
where we defined the evolution 
\begin{equation}\label{eq:U}
U(t) = e^{-it \left( -\frac{1}{2} \Delta_X - 2 \Delta_r + (1/h^2) V (r/h) \right)} = : U_X (t) \cdot U_r (t) \end{equation}
with $U_X (t) = e^{i (t/2) \Delta_X}$ denoting the free evolution in the center of mass coordinate $X$ and \[ U_r (t) = e^{-it \left( -2 \Delta_r + (1/h^2) V(r/h) \right)}  \] the evolution in the relative coordinate $r$. Observe that \begin{equation}\label{eq:Ur-alpha0} U_r (t) \alpha_0 (r/h) = e^{itE_b/h^2} \alpha_0 (r/h) \end{equation}

\medskip

We are now ready to state our main results.
\begin{theorem}\label{thm:main}
Suppose Condition 1 above is satisfied, and assume $W \in H^1 (\bR^3) \cap L^\infty (\bR^3)$. Consider an initial BCS state $0 \leq \Gamma_0 \leq 1$ with $ \tr \, \gamma_0 \leq C/h$, and
\begin{equation}\label{eq:en-ini} \cE_{BCS} (\Gamma_0) \leq - \frac{E_b}{2} \, \tr \, \gamma_0 + C h \end{equation}
Let $\wt{\alpha}_t$ be the solution of the integral equation (\ref{eq:wtalpha-duh}), and set, for arbitrary $t \in \bR$, 
\begin{equation}\label{eq:wtpsi} 
\psi_t (X) := \frac{e^{i t E_b / h^2}}{h} \int dr \, \alpha_0 (r/h) \wt{\alpha}_t (X,r) \end{equation}
Then we have $\| \psi_t \|_{H^1} \leq C$, uniformly in $h$ and in $t \in \bR$, and 
\[ \wt{\alpha}_t (X,r) = \frac{e^{-it E_b /h^2}}{h^2} \psi_t (X) \alpha_0 (r/h) + e^{-it E_b /h^2} \wt{\xi}_t (X,r) \] with $\| \wt{\xi}_t \|^2_{L^2 
(\bR^3 \times \bR^3)} \leq C h$ and $\| \wt{\xi}_t \|^2_{H^1 (\bR^3 \times \bR^3)} \leq C h^{-1}$. Moreover, if $\ph_t (X)$ denotes the solution of the Gross-Pitaevskii equation
\begin{equation}\label{eq:GP} i\partial_t \ph_t  = -\frac{1}{2} \Delta \ph_t + 2 W \ph_t + 2 g |\ph_t|^2 \ph_t \end{equation}
with the initial data $\ph_{t=0} (X) = \psi_0 (X)$ and with the coupling constant $g$ as in (\ref{eq:g}), we have
\begin{equation}\label{eq:conv} \| \psi_t - \ph_t \|_2 \leq C h^{1/2} e^{c |t|} \end{equation}
for constants $C,c >0$.  
\end{theorem} 

{\bf Remarks: }
\begin{itemize}
\item The theorem implies that, in good approximation, the pair density $\alpha_t$ is given by the product of the ground state wave function $\alpha_0$ (varying on the microscopic scales, describing the internal degrees of freedom of the fermion pairs and remaining constant in time) and the solution of the Gross-Pitaevskii equation, varying on the macroscopic scale and describing the distributions of the pairs in space. 
\item We only assume that the number of particles $\tr \gamma_0$ is bounded above by $C/h$. 
The interesting regime is the one with $\tr \gamma_0$ of the order $1/h$. It is easy to see that $\| \psi_t \|_2 \simeq h \tr \, \gamma_0$; hence, if $\tr \, \gamma_0 \ll 1/h$, the convergence (\ref{eq:conv}) is less interesting (it becomes trivial if $\tr \, \gamma_0 \leq 1/h^{1/2}$).  
\item The Gross-Pitaevskii equation (\ref{eq:GP}) is a defocusing non-linear Schr\"odinger equation. It is known to be globally well posed in the space $H^1 (\bR^3)$. In other words, for any initial data $\ph_0 \in H^1 (\bR^3)$, there is a unique solution $\ph_t \in C((-\infty, \infty), H^1 (\bR^3))$. Since the nonlinearity is defocusing and the external potential is bounded, it is easy to see that the $H^1$ norm of $\ph_t$ is uniformly bounded in time; i.e., there exists a constant $C$, depending only on the $H^1$-norm of the initial data, with $\| \ph_t \|_{H^1} \leq C$ for all $t \in \bR$. 
\item The fact that $\psi_t$ represents a pair of fermions is nicely reflected in the Gross-Pitaevskii
equation. The factor $1/2$ in front of the kinetic term corresponds to the total mass $m=1$ of two fermions (we use units where each fermion has mass $1/2$). The factor $2$ in front of the external potential $W$ corresponds to the total charge $e=2$ of the two fermions. 
\item Our analysis can also be extended to particles in small external magnetic fields, varying on the macroscopic scale. The magnetic field acts on the particle through a minimally coupled magnetic potential, which, in macroscopic units, has the form $h A(x) = h (A_1 (x) , A_2 (x) , A_3 (x))$.  The limiting Gross-Pitaevskii equation (\ref{eq:GP}) is replaced in this case by 
\begin{equation}\label{eq:GP-A}  i\partial_t \ph_t  = \frac{1}{2} (-i \nabla - A)^2  \ph_t + 2 W \ph_t + 2 g |\ph_t|^2 \ph_t \end{equation}
The proof of Theorem \ref{thm:main} extends to a derivation of (\ref{eq:GP-A}), under the assumptions that $A$ is continuous and decays at infinity and that the spectrum of $(-i\nabla - A)^2$ is absolutely continuous (these conditions guarantee the validity of Stricharzt estimates similar to those discussed in Appendix \ref{sec:strich} for the evolution generated by the magnetic Laplacian; see \cite{G}).
\end{itemize}

A crucial ingredient of our analysis is the following energy bound established in \cite{FHSS,HS}, which we adapt here to our setting. 
\begin{proposition}\label{prop}
Suppose Condition 1 is satisfied, and $W \in L^\infty (\bR^3)$, and consider a BCS state $0 \leq \Gamma \leq 1$ with $\tr \, \gamma \leq C/h$ and 
\[ \cE_{BCS} (\Gamma) \leq -\frac{E_b}{2} \,  \tr \, \gamma  + C h.\]
Then $\tr \gamma^2 \leq \tr  (\gamma - \bar{\alpha} \alpha) \leq C h$ for $C> 0$ independent of $h$ (from $\gamma (1-\gamma) \geq \bar{\alpha} \alpha$ we find that $\gamma - \bar{\alpha} \alpha \geq \gamma^2$ and in particular that $\gamma-\bar{\alpha} \alpha$ is a non-negative operator). Moreover, with
\begin{equation}\label{eq:psi-def} \psi (X) := \frac{1}{h} \int dr \alpha_0 \left(r/h \right) \, \wt{\alpha} (X,r) \end{equation}
we have $\| \psi \|_{H^1} \leq C$, uniformly in $h > 0$, and 
\[ \wt{\alpha} (X,r) = \frac{1}{h^2} \psi (X) \alpha_0 (r/h) + \wt{\xi} (X,r) \]
with the bound $\| \wt{\xi} \|^2_{L^2} \leq C h$ and $\| \wt{\xi} \|^2_{H^1} \leq C h^{-1}$. In particular, since the number of particles $\tr \, \gamma_t$ and the 
BCS-energy $\cE_{BCS} (\Gamma_t)$ are preserved by the BCS time-evolution, the solution $\Gamma_t$ of the time dependent BCS equation with an initial BCS state $0 \leq \Gamma_0\leq 1$ with $\tr \, \gamma_0 \leq C/h$ and 
\[ \cE_{\text{BCS}} (\Gamma_0) \leq -\frac{E_b}{2} \tr \, \gamma_0 + C h \]
is such that $\tr \, \gamma_t^2 \leq \tr \, ( \gamma_t - \bar{\alpha}_t \alpha_t ) \leq C h$ and such that, with the definition 
\[ \psi_t (X) := \frac{e^{itE_b/h^2}}{h} \int dr \alpha_0 (r/h) \wt{\alpha}_t (X,r) \, , \]
one has $\| \psi_t \|_{H^1} \leq C$, uniformly in $h$ and in $t$, and
\[ \wt{\alpha}_t (X,r) = \frac{e^{-itE_b/h^2}}{h^2} \psi_t (X) \alpha_0 (r/h) + e^{-it E_b/h^2} \wt{\xi}_t (X,r) \]
with the error $\wt{\xi}_t$ satisfying $\| \wt{\xi}_t \|_{L^2}^2 \leq C h$ and $\| \wt{\xi}_t \|_{H^1}^2 \leq C h^{-1}$. 
\end{proposition}
{\bf Remark:} We will also use the notation $\xi_s (x,y) = \wt{\xi}_s ((x+y)/2, x-y)$ (so that $\alpha_s (x,y) = (1/h^2) \psi_s ((x+y)/2) \alpha_0 ((x-y)/h) + \xi_s (x,y)$). The bounds for $\wt{\xi}_s$ immediately imply similar bounds for $\xi_s$.
\begin{proof}
We start by noticing that
\[ \begin{split} C h \geq \; & \cE_{\text{BCS}} (\Gamma) + \frac{E_b}{2} \tr \gamma \\ = \; & \tr \left( - h^2 \Delta + h^2 W + \frac{E_b}{2} \right) \gamma + \frac{1}{2} \int dx dy \, V ((x-y)/h) |\alpha (x,y)|^2 \\ \geq \; & \tr \left( -h^2 \Delta + \frac{E_b}{2} \right) \gamma + \frac{1}{2} \int dx dy \, V ((x-y)/h) |\alpha (x,y)|^2 - h^2 \| W \|_\infty \tr \, \gamma 
\end{split}  \]
Since $\tr \gamma \leq C /h$, and since $\gamma - \bar{\alpha} \alpha \geq \gamma^2 \geq 0$, 
we find
\[ \begin{split} 
Ch &\geq  \tr \left( -h^2 \Delta + \frac{E_b}{2} \right) (\gamma - \bar{\alpha} \alpha) +  \tr \left( -h^2 \Delta + \frac{E_b}{2} \right) \overline{\alpha} \alpha + \frac{1}{2} \int dx dy \, V ((x-y)/h) |\alpha (x,y)|^2  \\
&\geq  C \, \tr \, (\gamma - \bar{\alpha} \alpha) +  \tr \left( -h^2 \Delta + \frac{E_b}{2} \right) \overline{\alpha} \alpha + \frac{1}{2} \int dx dy \, V ((x-y)/h) |\alpha (x,y)|^2
\end{split} \]
Next we introduce center of mass coordinate $X= (x+y)/2$ and relative coordinate $r = x-y$. 
Since $\Delta_x + \Delta_y = (1/2) \Delta_X + 2 \Delta_r$, we conclude that
\[\begin{split}  \tr \Delta \overline{\alpha} \alpha & = \int dx dy \, \overline{\Delta_x \alpha (x,y)} \, \alpha (x,y) =
 \frac{1}{2} \int dx dy \, \overline{ (\Delta_x + \Delta_y) \alpha (x,y)} \alpha (x,y) \\ &= \int dX dr \, \overline{ \left(\frac{1}{4} \Delta_X  + \Delta_r \right) \wt{\alpha} (X,r)} \, \wt{\alpha} (X,r) \end{split} \]
 Therefore
 \begin{equation}\label{eq:prop-1} \begin{split} Ch \geq \; &C \tr \,  (\gamma - \bar{\alpha} \alpha)  +  \frac{h^2}{4} \int dX dr \, \left| \nabla_X \wt{\alpha} (X,r) \right|^2
\\ &+ \int dX \left\langle \wt{\alpha} (X,.) , \left(-h^2 \Delta_r + \frac{1}{2} V( r/h) + \frac{E_b}{2} \right) \wt{\alpha} (X,.) \right\rangle \end{split} \end{equation}
This implies first of all the bound $\tr \, \gamma^2 \leq \tr (\gamma - \bar{\alpha} \alpha) \leq C h$ (the first inequality is a general property of BCS states, since $\gamma (1-\gamma) \geq \bar{\alpha} \alpha$). Now we set
\begin{equation}\label{eq:prop-2} \wt{\alpha} (X,r) = \frac{1}{h^2} \psi (X) \alpha_0 (r/h) + \wt{\xi} (X,r) \end{equation}
with $\psi$ as defined in (\ref{eq:psi-def}). {F}rom (\ref{eq:prop-1}) (using that $(-\Delta + (1/2) V) \alpha_0 = (E_b/2) \alpha_0$), we find 
 \begin{equation}\label{eq:prop-3} 
 \begin{split}
Ch \geq \; & C \tr \, (\gamma-\bar{\alpha} \alpha) +  \frac{h^2}{4} \int dX dr \, \left| \nabla_X \wt{\alpha} (X,r) \right|^2
\\ &+ \int dX \left\langle \wt{\xi} (X,.) , \left(-h^2 \Delta_r + \frac{1}{2} V( r/h) + \frac{E_b}{2} \right) \wt{\xi} (X,.) \right\rangle \end{split}
\end{equation}
Eq. (\ref{eq:prop-2}) implies that, for (almost) every $X \in \bR^3$, $\wt{\xi} (X,.)$ is orthogonal to $\alpha_0$. Hence, if $\kappa > 0$ denotes the spectral gap between the ground state of $-\Delta + (1/2) V$ and its first excited state, we deduce from (\ref{eq:prop-3}) that 
\[ \int dX dr |\wt{\xi} (X,r)|^2 \leq C \kappa^{-1} h \]
The $L^2$-norm of $\psi$ can be estimated using the very definition (\ref{eq:psi-def}). By Cauchy-Schwarz, we find
\[ \begin{split}  \| \psi \|_2^2 \leq \; &\frac{1}{h^2} \int dX dr_1 dr_2 \alpha_0 (r_1/h) \alpha_0 (r_2/h) \overline{\wt{\alpha} (X,r_1)} \wt{\alpha} (X,r_2)\\  \leq \; &\frac{1}{h^2} \int dX dr_1 dr_2 |\alpha_0 (r_2/h)|^2 |\wt{\alpha} (X,r_1)|^2 \leq h \| \wt{\alpha} \|_2^2 \leq C \end{split} \]
because $\| \wt{\alpha} \|_2^2 = \tr \, \overline{\alpha} \alpha \leq \tr \gamma \leq C/h$. To bound the $L^2$-norm of $\nabla \psi$ we use again Cauchy-Schwarz to estimate
\[ \begin{split} \| \nabla \psi \|_2^2 &\leq \frac{1}{h^2} \int dX dr_1 dr_2 \alpha_0 (r_1/h) \alpha_0 (r_2/h) \overline{\nabla_X \wt{\alpha} (X,r_1)} \, \nabla_X \wt{\alpha} (X,r_2) \\ &\leq \frac{1}{h^2} \int dX dr_1 dr_2 |\alpha_0 (r_2/h)|^2 |\nabla_X \wt{\alpha} (X,r_1)|^2 \leq h \| \nabla_X \wt{\alpha} \|_2^2 \leq C \end{split} \]
where we used (\ref{eq:prop-3}) to bound $\| \nabla_X \wt{\alpha} \|_2^2 \leq C/h$. Since 
\[ \left \| \nabla_X  \left[ \frac{1}{h^2} \psi (X) \alpha_0 (r/h) \right] \right\|_2^2 = \frac{1}{h^4} \int dX dr \, |\nabla \psi (X)|^2 \alpha_0^2 (r/h) = \frac{\| \alpha_0 \|_2^2 \| \nabla \psi \|_2^2}{h} \leq C h^{-1} \]
the bound $\| \nabla_X \wt{\alpha} \|_2^2 \leq C/h$ also implies that $\| \nabla_X \wt{\xi} \|_2^2 \leq C h^{-1}$. On the other hand, since $V \in L^\infty (\bR^3)$, we have 
$-2h^2 \Delta_r + V(r/h) \geq - 2 h^2 \Delta_r - C$ and therefore, by (\ref{eq:prop-3}),  
\[ C h \geq \int dX \, \left\langle \wt{\xi} (X,.) , \left[- 2h^2 \Delta_r - C \right] \wt{\xi} (X,.) \right\rangle = 2 h^2 \int dX dr \, |\nabla_r \wt{\xi} (X,r)|^2 - C \| \wt{\xi} \|_2^2 \]
Since $\| \wt{\xi} \|_2^2 \leq C h$, we immediately conclude that $\| \nabla_r \, \wt{\xi} \|_2^2 \leq C h^{-1}$.
\end{proof}

\section{Proof of Theorem \ref{thm:main}}
\label{sec:proof}

We start from the definition (\ref{eq:wtpsi}) for the wave function $\psi_t$, that is 
$$\psi_t (X) := \frac{e^{i t E_b / h^2}}{h} \int dr \, \alpha_0 (r/h) \wt{\alpha}_t (X,r),$$
and we insert the integral equation (\ref{eq:wtalpha-duh}) for $\wt{\alpha}_t$. We find
\[ \begin{split} 
\psi_t &(X) \\ = \; & \frac{e^{iE_b t/h^2}}{h} \int dr \alpha_0 (r/h) U(t) \wt{\alpha}_0 (X,r)  \\ &+  \frac{e^{i t E_b / h^2}}{ih} \int dr \, \alpha_0 (r/h) \int_0^t ds \, U(t-s) (W (X+r/2) + W(X-r/2)) \wt{\alpha}_s (X,r) 
\\ &- \frac{e^{itE_b/h^2}}{ih^3} \int dr \alpha_0 (r/h) \int_0^t ds \, U(t-s) \int dz \, \gamma_s (X+r/2, X+z -r/2) V(z/h) \\ &\hspace{8cm} \times \wt{\alpha}_s (X-r/2 + z/2,z) \\ &- \frac{e^{itE_b/h^2}}{ih^3} \int dr \alpha_0 (r/h) \int_0^t ds \, U(t-s) \int dz \gamma_s (X -r/2, X-z+r/2) V(z/h)  \\ &\hspace{8cm} \times \wt{\alpha}_s (X+r/2-z/2,z) \end{split} \]
We use now the decomposition (\ref{eq:U}) of the evolution $U(t)$, saying that $U(t) = U_X(t) U_r(t)$. Since we integrate over $r$ we can let $U_r(t)$ act on the left in the form of its adjoint and use (\ref{eq:Ur-alpha0}), more precisely 
$U^\ast_r(t-s) \alpha_0 = e^{-i(t-s)E_b/h^2} \alpha_0$, such that the phase in the $t$-variable cancels. We conclude that
\begin{equation}\label{eq:wtpsi-123}  \begin{split} 
\psi_t &(X) \\  = \; &U_X (t)  \psi_0 (X)  + \frac{1}{ih} \int_0^t ds \, U_X (t-s) \int dr \, \alpha_0 (r/h) 
(W (X+r/2) + W(X-r/2)) e^{is E_b/h^2} \wt{\alpha}_s (X,r) 
\\ & -  \frac{1}{ih^3}  \int_0^t ds \, U_X (t-s)  \int dr dz  \, \alpha_0 (r/h) \gamma_s (X+r/2, X+z -r/2) V(z/h) \\ & \hspace{8cm} \times e^{is E_b/h^2} \wt{\alpha}_s (X+(z-r)/2,z) \\ &- \frac{1}{ih^3}  \int_0^t ds \, U_X (t-s)  \int dr dz \, \alpha_0 (r/h) \gamma_s (X -r/2, X-z+r/2) V(z/h) \\ & \hspace{8cm} \times e^{is E_b/h^2} \wt{\alpha}_s (X+ (r-z)/2,z) \end{split} .\end{equation}
Notice that we can see already the Gross-Pitaevskii equation appear. To leading order the integration over the function $\alpha_0(r/h)$ forces $r$ to be close to $0$ in the above integrals, such that one can see the term with the potential $2W(X)$ emerge. For the non-linear part we will simply use the fact that $\gamma_s$ is essentially $\overline{\alpha_s } \alpha_s$ such that the last term has to include $|\psi|^2 \psi$ to leading order. 

By Proposition \ref{prop}, the second term on the r.h.s. of the last equation \eqref{eq:wtpsi-123} can be decomposed as follows:
\[ \begin{split} 
 \frac{1}{h} \int dr \, \alpha_0 (r/h)  (W &(X+r/2) + W(X-r/2)) e^{is E_b/h^2} \wt{\alpha}_s (X,r)  \\ 
= \; & \frac{\psi_s (X)}{h^3} \int dr \alpha^2_0 (r/h) (W(X+r/2) + W(X-r/2)) \\ &+ \frac{1}{h} \int dr \alpha_0 (r/h) (W (X+r/2) + W(X-r/2)) \wt{\xi}_s (X,s) \\
= \; &2 W(X) \psi_s (X) \\ &+ \psi_s (X) \int dr \alpha^2_0 (r) (W(X+hr/2) - W(X)) 
\\ & + \psi_s (X) \int dr \alpha^2_0 (r) (W(X-hr/2) - W(X)) 
\\ &+ \frac{1}{h} \int dr \alpha_0 (r/h) (W (X+r/2) + W(X-r/2)) \wt{\xi}_s (X,r) \\
=: \; & 2 W(X) \psi_s (X) + G^{(0)}_{1,s} (X) + G^{(0)}_{2,s} (X) + G^{(0)}_{3,s} (X) 
\end{split} \]
where we used the normalization $\| \alpha_0 \|_2 = 1$. We observe that
\[\begin{split}  |G^{(0)}_{1,s} (X)| \leq \; &(h/2) \, |\psi_s (X)|  \int dr \alpha^2_0 (r) \left| \int_0^1 d\kappa \, r \cdot \nabla W(X+h\kappa r/2) \right| \\ \leq \; &(h/2) |\psi_s (X)| \int dr \int_0^1 d\kappa \, \alpha_0 (r)^2 \, |r| \,  |\nabla W (X+h\kappa r/2)| \end{split} \]
Therefore
\[ \begin{split} \| G^{(0)}_{1,s} \|_{6/5}^{6/5}  \leq \; &C h^{6/5}  \int dX |\psi_s (X)|^{6/5}  \left| \int dr \int_0^1 d\kappa |r| \alpha_0^2 (r)  |\nabla W (X+h\kappa r/2)| \right|^{6/5}  \\
\leq \; & C h^{6/5} \left( \int dr  |r| \alpha_0^{2} (r) \right)^{1/5}  \, \int dX dr \int_0^1 d\kappa  |r| \alpha^2_0 (r) |\psi_s (X)|^{6/5}  |\nabla W (X+h\kappa r/2)|^{6/5} \end{split}\]
where we used H\"older inequality with exponents $6$ and $6/5$ (with respect to the measure $d\kappa dr |r| \alpha^2 (r)$). {F}rom the assumptions on $W$, Condition 1 and Proposition \ref{prop} (which implies that $\| \psi_s \|_{H^1}$ is uniformly bounded in $h$ and in $s$), we find 
\[ \begin{split} 
\| G^{(0)}_{1,s} \|_{6/5}^{6/5}  
\leq \; & C h^{6/5} \left[\int dX dr \, 
|r| \alpha^2_0 (r) \, |\psi_s (X)|^3 \right]^{2/5} \left[ \int dX dr \,
|r| \alpha^2_0 (r) \,  |\nabla W (X)|^2 \right]^{3/5} \\
\leq \; &C h^{6/5} \| \psi_s \|_{H^1}^{6/5} \| \nabla W \|_2^{6/5} \leq C h^{6/5} \end{split} \]
The term $G^{(0)}_{2,s}$ can be bounded analogously. To bound $G^{(0)}_{3,s}$, on the other hand, we observe that
\[ \begin{split}  \| G^{(0)}_{3,s} \|_{6/5}^{6/5} \leq \; & \frac{1}{h^{6/5}} \int dX \left| \int dr \alpha_0 (r/h) (W (X+r/2) + W(X-r/2)) \wt{\xi}_s (X,r) \right|^{6/5} \\ \leq \; & \frac{1}{h^{6/5}} \left[ \int dr \alpha_0 (r/h) \right]^{1/5}  \, \int dX dr \alpha_0 (r/h) |W (X+r/2) + W(X-r/2)|^{6/5}  |\wt{\xi}_s (X,r)|^{6/5} \\ 
\leq \; &\frac{C}{h^{3/5}} \left[ \int dX dr \alpha_0 (r/h) |W (X+r/2) + W(X-r/2)|^{3} \right]^{2/5} \\ &\hspace{5cm} \times 
\left[ \int dX dr \alpha_0 (r/h) |\wt{\xi}_s (X,r)|^2 \right]^{3/5} \\
\leq &C h^{3/5}  \| W \|_{3}^{6/5} \| \wt{\xi}_s \|_2^{6/5} \leq C h^{6/5} 
\end{split} \]
for a constant $C$ depending on $\| \alpha_0 \|_1$, on $\| \alpha_0 \|_\infty$, and on $\| W \|_{H^1}$. Here we used again Proposition \ref{prop} to bound $\| \wt{\xi}_s \|_2 \leq C h^{1/2}$ and the assumptions on $W$. We conclude that
\begin{equation}\label{eq:W-term}
\begin{split}
\frac{1}{h} \int dr \, \alpha_0 (r/h) (W (X+r/2) + W(X-r/2)) &e^{is E_b/h^2} \wt{\alpha}_s (X,r)  = 2 W(X) \psi_s (X) + \sum_{j=1}^3 G^{(0)}_{j,s} (X)  \end{split} \end{equation}
where $\| G^{(0)}_{j,s} \|_{6/5} \leq C h$, for all $j=1,2,3$ and for a constant $C>0$ independent of $h$ and of time $s$. We consider now the third term on the r.h.s. of (\ref{eq:wtpsi-123}). Using Proposition \ref{prop}, we find
\[ \begin{split}  A := \; &-\frac{1}{h^3} \int dr dz \alpha_0 (r/h)  \gamma_s (X+r/2, X+z -r/2)  V(z/h) e^{i s E_b/h^2} \wt{\alpha}_s (X + (z-r)/2, z) \\  = \; & -\frac{1}{h^5} \int dr dz  \alpha_0 (r/h)  \gamma_s (X+r/2, X+z -r/2) (V\alpha_0) (z/h) \psi_s ( X+(z-r)/2)  \\
&-\frac{1}{h^3} \int dr dz \, \alpha_0 (r/h) \gamma_s (X+r/2, X+z -r/2) V (z/h)e^{is E_b/h^2} \wt{\xi}_s (X+ (z-r)/2, z) \end{split} \]
Next, we write $\gamma_s = \bar{\alpha}_s \alpha_s + (\gamma_s - \bar{\alpha}_s \alpha_s)$. Here $\gamma_s - \bar{\alpha}_s \alpha_s$ is a positive operator (since $\gamma_s (1-\gamma_s) \geq \bar{\alpha}_s \alpha_s$ implies that $\gamma_s - \bar{\alpha}_s \alpha_s \geq \gamma_s^2$), with $\tr \, (\gamma_s - \bar{\alpha}_s \alpha_s) \leq C h$ (compared with $\tr \, \bar{\alpha}_s \alpha_s \simeq C h^{-1}$); we will use this fact to show, in Lemma \ref{lm:Gj}, that the contribution from $ (\gamma_s - \bar{\alpha}_s \alpha_s)$ is small, in the limit $h\to 0$. We find
\[ \begin{split} A = \; & -\frac{1}{h^5} \int dr dz dw \, \alpha_0 (r/h) \overline{\alpha}_s (X+r/2, X+r/2 + w) \alpha_s (X+r/2 + w, X + z -r/2) (V\alpha_0) (z/h) \\ & \hspace{7cm} \times  \psi_s ( X+(z-r)/2) \\
\; & -\frac{1}{h^5} \int dr dz \, \alpha_0 (r/h) (\gamma_s - \bar{\alpha}_s \alpha_s) ( X+r/2, X+z - r/2) V(z/h)  \psi_s ( X+(z-r)/2) \alpha_0 (z/h) \\ &-\frac{1}{h^3} \int dr dz \, \alpha_0 (r/h)  \gamma_s (X+r/2, X+z -r/2) V(z/h) e^{is E_b/h^2}\wt{\xi}_s (X+ (z-r)/2, z) \end{split} \]
We express again the kernels $\alpha_s$ using center of mass and relative coordinates. We find
\[ \begin{split}  
A =  \; & -\frac{1}{h^5} \int dr dz dw \, \alpha_0 (r/h) \overline{\wt{\alpha}_s (X+(r+w)/2 , w)} \wt{\alpha}_s (X + (z+w)/2,  r + w - z) (V\alpha_0) (z/h) \\ &\hspace{7cm} \times \psi_s ( X+(z-r)/2) \\
\; & -\frac{1}{h^5} \int dr  dz \, \alpha_0 (r/h) (\gamma_s - \bar{\alpha}_s \alpha_s) ( X+r/2, X+z - r/2) V(z/h) \psi_s ( X+(z-r)/2) \alpha_0 (z/h) \\ &-\frac{1}{h^3} \int dr dz \,  \alpha_0 (r/h)  \gamma_s (X+r/2, X+z -r/2) V(z/h)  e^{is E_b/h^2} \wt{\xi}_s (X+ (z-r)/2, z) \end{split} \]
Finally, we use again Proposition \ref{prop} to rewrite the kernels $\wt{\alpha}_s$.
We find that
\[ \begin{split}
A = \; & -\frac{1}{h^7} \int dr dzdw \, \alpha_0 (r/h) e^{is E_b/h^2} \overline{\psi}_s (X+(r +z)/2) \alpha_0 (w/h)  \wt{\alpha}_s (X + (z+w)/2,  r + w - z) \\ &\hspace{7cm} \times (V\alpha_0) (z/h)  \psi_s ( X+(z-r)/2) \\ &-\frac{1}{h^5} \int dr dz dw \,  \alpha_0 (r/h) e^{isE_b/h^2} \overline{\wt{\xi}_s (X+(r+w)/2 , w)} \wt{\alpha}_s (X + (z+w)/2,  r + w - z) \\ &\hspace{7cm} \times (V\alpha_0) (z/h) \psi_s ( X+(z-r)/2) \\
\; & -\frac{1}{h^5} \int dr dz \,  \alpha_0 (r/h)(\gamma_s - \bar{\alpha}_s \alpha_s) ( X+r/2, X+z - r/2) (V\alpha_0) (z/h) \psi_s ( X+(z-r)/2) \\ &-\frac{1}{h^3} \int dr dz \, \alpha_0 (r/h)  \gamma_s (X+r/2, X+z -r/2)  V(z/h)  \wt{\xi}_s (X+ (z-r)/2, z) \end{split} \]
and we conclude
\begin{equation}\label{eq:error-1} 
\begin{split}  
A = \; & -\frac{1}{h^9} \int dr dz dw \, \alpha_0 (r/h) \alpha_0 (w/h)  \overline{\psi}_s (X+(r+z)/2)  \psi_s (X + (z+w)/2) \alpha_0 ((r + w - z)/h) \\ &\hspace{7cm} \times  (V\alpha_0) (z/h)  \psi_s ( X+(z-r)/2) \\ 
 & -\frac{1}{h^7} \int dr dz dw \, \alpha_0 (r/h) \overline{\psi}_s (X+ (r+z)/2) \alpha_0 (w/h)  \wt{\xi}_s (X + (z+w)/2,  r + w - z)  \\ &\hspace{7cm} \times (V\alpha_0) (z/h) \psi_s ( X+(z-r)/2) \\ &-\frac{1}{h^7} \int dr dz dw \, \alpha_0 (r/h)  \overline{\wt{\xi}_s (X+(r+w)/2 , w)} \alpha_0 ((r+w-z)/h) \psi_s (X+(z+w)/2)
\\ & \hspace{7cm} \times  (V\alpha_0) (z/h) \psi_s ( X+(z-r)/2)   \\ &-\frac{1}{h^5} \int dr dz dw \, \alpha_0 (r/h)  \overline{\wt{\xi}_s (X+(r+w)/2 , w)}  \wt{\xi}_s (X + (z+w)/2,  r + w - z) \\ & \hspace{7cm} \times  (V\alpha_0) (z/h) \psi_s ( X+(z-r)/2)  \\
\; & -\frac{1}{h^5} \int drdz  \, \alpha_0 (r/h) (\gamma_s - \bar{\alpha}_s \alpha_s) ( X+r/2, X+z - r/2) (V\alpha_0) (z/h) \psi_s ( X+(z-r)/2)  \\ &-\frac{1}{h^3} \int dr dz \, \alpha_0 (r/h)  \gamma_s (X+r/2, X+z -r/2) V(z/h) \wt{\xi}_s (X+ (z-r)/2, z) \\
=: \; & M (X) + \sum_{j=1}^5 G^{(1)}_{j,s} (X) \end{split} \end{equation}
We decompose further the main term $M(X)$, writing 
\[ \begin{split} 
M(X) = \; & -\frac{1}{h^9} \int dr dz dw \, \alpha_0 (r/h)  \alpha_0 (w/h)  \alpha_0 ((r + w - z)/h) (V\alpha_0) (z/h)  \\ &\hspace{2cm} \times \overline{\psi}_s (X+r/2 +z/2) \psi_s (X + z/2 + w/2) \psi_s ( X-r/2+z/2) \\
= \; &- \int dr dz dw \, \alpha_0 (r) \alpha_0 (w) \alpha_0 (r+w-z) (V\alpha_0) (z) \\ &\hspace{2cm} \times 
\overline{\psi}_s (X+h(r+z)/2) \psi_s (X + h(z + w)/2) 
\psi_s (X +h (z- r)/2) \\ 
= \; &- \psi_s (X)  \int dr dz dw \, \alpha_0 (r) \alpha_0 (w) \alpha_0 (r+w-z) (V\alpha_0) (z)\\ &\hspace{2cm} \times \overline{\psi}_s (X+h(r+z)/2) \psi_s (X + h(z+w)/2) 
\\ &- \int dr dz dw \, \alpha_0 (r) \alpha_0 (w) \alpha_0 (r+w-z) (V\alpha_0) (z) \\ &\hspace{2cm} \times 
\overline{\psi}_s (X+h(r+z)/2) \psi_s (X + h(z+w)/2) \left( \psi_s (X+ h(z-r)/2) - \psi_s (X) \right)  \\ 
= \; &- \psi_s (X)^2  \int dr dz dw \, \alpha_0 (r) \alpha_0 (w) \alpha_0 (r+w-z) (V\alpha_0) (z) \overline{\psi}_s (X+h(r+z)/2)  \\
&- \psi_s (X)  \int dr dz dw \, \alpha_0 (r) \alpha_0 (w) \alpha_0 (r+w-z) (V\alpha_0) (z) 
\\ &\hspace{2cm} \times \overline{\psi}_s (X+h(r+z)/2) \left(\psi_s (X + h(z+w)/2) - \psi_s (X) \right)\\
 &- \int dr dz dw \, \alpha_0 (r) \alpha_0 (w) \alpha_0 (r+w-z) (V\alpha_0) (z) \\ &\hspace{2cm} \times
\overline{\psi}_s (X+h(r+z)/2) \psi_s (X + h(z+w)/2) 
\left( \psi_s (X+ h(z-r)/2) - \psi_s (X) \right)  \end{split}\]
In the first term, we isolate again the main contribution. We find
\[ \begin{split} 
M(X) = \; &- |\psi_s (X)|^2 \psi_s (X) \, \int dr dz dw \, \alpha_0 (r) \alpha_0 (w) \alpha_0 (r+w-z) (V\alpha_0) (z) \\
&- \psi_s (X)^2  \int dr dz dw \, \alpha_0 (r) \alpha_0 (w) \alpha_0 (r+w-z) (V\alpha_0) (z) 
\left(\overline{\psi}_s (X+h(r+z)/2) - \overline{\psi}_s (X) \right) \\
&- \psi_s (X)   \int dr dz dw \, \alpha_0 (r) \alpha_0 (w) \alpha_0 (r+w-z) (V\alpha_0) (z) 
\\ &\hspace{2cm} \times \overline{\psi}_s (X+h(r+z)/2) \left(\psi_s (X + h(z+w)/2) - \psi_s (X) \right)\\  &- \int dr dz dw \, \alpha_0 (r) \alpha_0 (w) \alpha_0 (r+w-z) (V\alpha_0) (z) 
\\ &\hspace{2cm} \times \overline{\psi}_s (X+h(r+z)/2) \psi_s (X + h(z+w)/2) 
\left( \psi_s (X + h(z-r)/2) - \psi_s (X) \right) \end{split} \]
We observe that
\[\begin{split} - \int dr &dz dw \, \alpha_0 (r) \alpha_0 (w) \alpha_0 (r+w-z) (V\alpha_0) (z)  \\ & =  - \int dz \left( \alpha_0 * \alpha_0 * \alpha_0 \right) (z)  \, (V\alpha_0) (z)  = \frac{1}{(2\pi)^3} \int dp \, |\widehat{\alpha}_0 (p)|^4 (2 p^2 + E_b) = g \end{split} \] where we used the eigenvalue equation $(-2 \Delta + V) \alpha_0 = - E_b \alpha_0$ to write $V\alpha_0 = - E_b \alpha_0 + 2 \Delta \alpha_0$ and therefore $\widehat{V\alpha_0} (p) = - (E_b+2p^2) \widehat{\alpha}_0 (p)$. We conclude that 
\[ M (X) = g |\psi (X)|^2 \psi (X) + \sum_{j=1}^3 G^{(2)}_{j,s} (X) \]
with the three error terms
\begin{equation}\label{eq:G0j} \begin{split}  G^{(2)}_{1,s} (X) &= - \psi_s (X)^2  \int dr dz dw \, \alpha_0 (r) \alpha_0 (w) \alpha_0 (r+w-z) (V\alpha_0) (z) 
\left(\overline{\psi}_s (X+h(r+z)/2) - \overline{\psi}_s (X) \right) \\
G^{(2)}_{2,s} (X) &= - \psi_s (X)   \int dr dz dw \, \alpha_0 (r) \alpha_0 (w) \alpha_0 (r+w-z) (V\alpha_0) (z) 
\\ &\hspace{2cm} \times \overline{\psi}_s (X+h(r+z)/2) \left(\psi_s (X + h(z+w)/2) - \psi_s (X) \right) \\
G^{(2)}_{3,s} (X) &=  - \int dr dz dw \, \alpha_0 (r) \alpha_0 (w) \alpha_0 (r+w-z) (V\alpha_0) (z) 
\\ &\hspace{2cm} \times \overline{\psi}_s (X+h(r+z)/2) \psi_s (X + h(z+w)/2) 
\left( \psi_s (X + h(z-r)/2) - \psi_s (X) \right)
\end{split} \end{equation}
Together with (\ref{eq:error-1}), this implies that
\[ \begin{split}  A &= -\frac{1}{h^3} \int dr dz \alpha_0 (r/h) \gamma_s (X + r/2, X+z-r/2)  V(z/h) e^{i s E_b/h^2} \wt{\alpha}_s (X-r/2 + z/2, z) \\ &= g \, |\psi_s (X)|^2 \psi_s (X) + \sum_{j=1}^5 G^{(1)}_{j,s} (X) + \sum_{j=1}^3 G^{(2)}_{j,s} (X) \end{split} \]
In Lemmas \ref{lm:G0j} and \ref{lm:Gj} below, we show that $\| G^{(1)}_{j,s} \|_{6/5} \leq C h$ for $j =1,2$, $\| G^{(1)}_{j,s} \|_{6/5} \leq C h^{1/2}$ for $j=3,5$, $\| G^{(1)}_{4,s} \|_{6/5} \leq C h^{3/2}$, and $\| G^{(2)}_{j,s} \|_2 \leq C h$ for all $j=1,2,3$, for a constant $C$ independent of $h$ and of $s$.  In particular, $\| G^{(i)}_{j,s} \|_{6/5} \leq C h^{1/2}$ for all $i=1,2$ and $j$'s. Similarly, one can write the last term on the r.h.s. of (\ref{eq:wtpsi-123}) as 
\[ \begin{split} -\frac{1}{h^3} \int dr dz \alpha_0 (r/h) & \gamma_s (X - r/2, X-z+r/2)  V(z/h) e^{i s E_b/h^2} \wt{\alpha}_s (X+r/2 - z/2, z) \\ &= g |\psi_s (X)|^2 \psi_s (X) + \sum_{j=1}^5 G^{(3)}_{j,s} (X) + \sum_{j=1}^3 G^{(4)}_{j,s} (X) \end{split} \]
where $\| G^{(i)}_{j,s} \|_{6/5} \leq C h^{1/2}$ for $i=3,4$ and all $j$'s, and for a constant $C$ independent of $h$ and of $s$. Therefore (\ref{eq:wtpsi-123}) and (\ref{eq:W-term}) imply that 
\[  \begin{split} 
\psi_t (X) = \; &U_X (t)  \psi_0 (X)  + \frac{2}{i} \int_0^t ds \, U_X (t-s) W(X) \psi_s (X) \\ &+ \frac{2 g}{i} \int_0^t \, ds \, U_X (t-s) |\psi_s (X)|^2 \psi_s (X) + \frac{1}{i} \int_0^t ds \, U_X (t-s) F (s) \end{split} \]
with the error term $F = \sum_{j=1}^3 G^{(0)}_{j,s} + \sum_{j=1}^{5} (G^{(1)}_{j,s} + G^{(3)}_{j,s}) + \sum_{j=1}^3 (G^{(2)}_{j,s} +G^{(4)}_{j,s})$ is such that $\| F (s) \|_{6/5} \leq C h^{1/2}$. 
We compare $\psi_t (X)$ with the solution of the Gross-Pitaevskii equation (\ref{eq:GP}), which can be rewritten in integral form as
 \[  \ph_t (X) = \; U_X (t)  \ph_0 (X) + \frac{2}{i} \int_0^t ds \, U_X (t-s) W(X) \ph_s (X)   + \frac{2 g}{i} \int_0^t \, ds U_X (t-s) |\ph_s (X)|^2 \ph_s (X) \]
The difference between $\psi_t$ and $\ph_t$ is therefore given by 
\[ \begin{split} \psi_t - \ph_t = \; &U_X (t) (\psi_0 - \ph_0) -2i \int_0^t ds \, U_X (t-s) W (\psi_s - \ph_s) \\ &-2ig \int_0^t ds \, U_X(t-s) \left( |\psi_s |^2 \psi_s - |\ph_s|^2 \ph_s \right) -i  \int_0^t ds \, U_X (t-s) F (s) \end{split} \]
We use Strichartz estimates (see Appendix \ref{sec:strich}) to bound this difference. We find (since we will later iterate the next bound, we keep track of the difference of the initial data; at the end we'll set $\psi_0 = \ph_0$) 
\[ \begin{split} \| \psi_t - &\ph_t \|_{L^\infty ([0,T] ; L^2 (\bR^3))} \\ \leq \; & \| \psi_0 -\ph_0\|_2 + 
2 \| W (\psi_s - \ph_s) \|_{L^2 ([0,T], L^{6/5} (\bR^3))}  + 2g 
\left\| |\psi_s|^2 \psi_s - |\ph_s|^2 \ph_s \right\|_{L^2 ([0,T], L^{6/5} (\bR^3))} \\ &
+ \| F \|_{L^2 ([0,T], L^{6/5} (\bR^3))} \\
\leq \; & \| \ph_0 -\psi_0 \|_2 + 2 \left( \int_0^T ds \| W (\psi_s - \ph_s) \|_{6/5}^2 \right)^{1/2} +
\left( \int_0^T ds \|  |\psi_s|^2 \psi_s - |\ph_s|^2 \ph_s \|^2_{6/5} \right)^{1/2} \\ & + \left( \int_0^T  ds \, \| F (s) \|_{6/5}^2 \right)^{1/2} \\
\leq \; & \| \ph_0 -\psi_0 \|_2 + 2 T^{1/2} \| W \|_3 \| \psi_s - \ph_s \|_{L^\infty ([0,T] ; L^2 (\bR^3))}
+ T^{1/2} \sup_{s \in [0,T]} \| |\psi_s|^2 \psi_s -|\ph_s|^2 \ph_s \|_{6/5} \\ &+ C T^{1/2} h^{1/2} 
 \end{split} \]
Next, we estimate the difference $\| |\psi_s|^2 \psi_s - |\ph_s|^2 \ph_s \|_{6/5}$ in terms of $\| \psi_s - \ph_s \|_2$, using the a-priori bounds for $\| \psi_s \|_{H^1}$ (given by Proposition \ref{prop}) and $\| \ph_s \|_{H^1}$ (see the remark after Theorem \ref{thm:main}).
\[ \begin{split}
\| \psi_t - &\ph_t \|_{L^\infty ([0,T] ; L^2 (\bR^3))} \\
\leq \; &  \| \ph_0 -\psi_0 \|_2 +  2 T^{1/2} \| W \|_3 \| \psi_s - \ph_s \|_{L^\infty ([0,T] ; L^2 (\bR^3))} +C T^{1/2} h^{1/2} \\ &+ T^{1/2} \sup_{s \in [0,T]} 
\left[ \| |\psi_s|^2 (\psi_s - \ph_s) \|_{6/5} + \| \psi_s \ph_s (\psi_s - \ph_s) \|_{6/5} + \| |\ph_s|^2 (\psi_s - \ph_s) \|_{6/5} \right] \\
\leq \;& \| \ph_0 - \psi_0 \|_2 + C T^{1/2}  h^{1/2} +  2 T^{1/2} \| W \|_3 \| \psi_s - \ph_s \|_{L^\infty ([0,T] ; L^2 (\bR^3))} \\ &+ T^{1/2} \sup_{s \in [0,T]}  \left( \| \psi_s \|^2_6 + \| \ph_s \|_6^2 \right)  \,\| \psi_s - \ph_s \|_2 \\ \leq \;&  \| \ph_0 -\psi_0 \|_2 + C T^{1/2}  h^{1/2} + C T^{1/2}  \| \psi_s - \ph_s \|_{L^\infty([0,T],L^2 (\bR^3))}  \end{split} \]
We choose $T_0 > 0$ such that the factor $CT^{1/2}$ appearing in front of the last term on the r.h.s. of the last equation is less than, say, $1/2$ (observe that $T_0$ depends only on the bounds for the $H^1$ norm of $\psi_s$ and of $\| W \|$ and is therefore independent of $h$ and $s$). Then we have
\[ \| \psi_s - \ph_s \|_{L^\infty([0,T],L^2)} \leq 2 \| \psi_0 - \ph_0\|_2 + Ch^{1/2} \]
for all $0\leq T \leq T_0$. For $t \in [nT_0, (n+1) T_0]$, we find analogously (since all bounds are uniform in time) 
\[ \begin{split} \| \psi_t - \ph_t\|_2 &\leq 2 \| \psi_{nT_0} - \ph_{nT_0} \|_2 + Ch^{1/2} \leq 4 \| \psi_{(n-1)T_0} - \ph_{(n-1)T_0} \|_2 +2 Ch^{1/2} + Ch^{1/2}  \\ &\leq  8 \| \psi_{(n-2)T_0} - \ph_{(n-2) T_0} \|_2 +4Ch^{1/2} + 2 Ch^{1/2} + Ch^{1/2} \leq 2^{n+1} \| \psi_0 - \ph_0 \|_2 + 2^{n+1} C h^{1/2} \end{split} \]
For initial data $\psi_0 = \ph_0$, we conclude that
\[ \| \psi_t - \ph_t \|_2 \leq C h^{1/2} e^{ct}  \]
for all $t \in \bR$. This concludes the proof of Theorem \ref{thm:main}.

\section{Control of error terms}

In this section, we estimate the size of the error terms from (\ref{eq:error-1}) and (\ref{eq:G0j}).

\begin{lemma} \label{lm:G0j}
Consider the error terms 
\begin{equation}\begin{split}  G^{(2)}_{1,s} (X) &= - \psi_s (X)^2  \int dr dz dw \, \alpha_0 (r) \alpha_0 (w) \alpha_0 (r+w-z) (V\alpha_0) (z) 
\left(\overline{\psi}_s (X+h(r+z)/2) - \overline{\psi}_s (X) \right) \\
G^{(2)}_{2,s} (X) &= - \psi_s (X)   \int dr dz dw \, \alpha_0 (r) \alpha_0 (w) \alpha_0 (r+w-z) (V\alpha_0) (z) 
\\ &\hspace{2cm} \times \overline{\psi}_s (X+h(r+z)/2) \left(\psi_s (X + h(z+w)/2) - \psi_s (X) \right) \\
G^{(2)}_{3,s} (X) &=  - \int dr dz dw \, \alpha_0 (r) \alpha_0 (w) \alpha_0 (r+w-z) (V\alpha_0) (z) 
\\ &\hspace{2cm} \times \overline{\psi}_s (X+h(r+z)/2) \psi_s (X + h(z+w)/2) 
\left( \psi_s (X + h(z-r)/2) - \psi_s (X) \right)
\end{split} \end{equation}
as defined in (\ref{eq:G0j}). Then $\| G_{0,j} \|_{6/5} \leq C h$, for a constant $C$ independent of $h$ and of the time $s$.
\end{lemma}

\begin{proof}
With \[ \begin{split} \psi_s (X+h(r+z)/2) - \psi_s (X) &= \int_0^1 d\kappa \, \frac{d}{d\kappa} \psi_s (X+h\kappa (r+z)/2)\\ & = h \int_0^1 d\kappa \, \nabla \psi_s (X+h(r+z)/2) \cdot (r+z)\end{split} \]
we find   
\[ \begin{split} |G^{(2)}_{1,s} (X)| \leq \; & h |\psi_s (X)|^2 \int dr dz dw \int_0^1 d\kappa \, |r+z| \, \alpha_0 (r) \alpha_0 (w) \alpha_0 (r+w-z) \\ & \hspace{5cm} \times  |(V\alpha_0) (z)| \, \left| \nabla \psi_s (X + \kappa h (r+z)/2) \right| \end{split} \]
Hence, using H\"older inequality with exponents $6$ and $6/5$ (with respect to the integration measure $dr dz dw d\kappa \, \alpha_0 (r) \alpha_0 (w) \alpha_0 (r+w-z) |(V\alpha_0)(z)| $), we find
\[ \begin{split} \| G^{(2)}_{1,s} \|_{6/5}^{6/5} \leq \; & h^{6/5}  \left( \int dr dw dz \int_0^1 d\kappa \, |r+z|^6 \alpha_0 (r) \alpha_0 (w) \alpha_0 (r+w-z) |(V\alpha_0)(z)| \right)^{1/5} \\ & \times  \int dX dr dz dw  \int_0^1 d\kappa  \, \alpha_0 (r) \alpha_0 (w) \alpha_0 (r+w-z) |(V\alpha_0)(z)| \\ &\hspace{3cm} \times   |\psi_s (X)|^{12/5} \left| \nabla \psi_s (X+\kappa h (r+z)/2) \right|^{6/5} 
\\ \leq \; & Ch^{6/5} \left[  \int dX dr dz dw \int_0^1 d\kappa \, \alpha_0 (r) \alpha_0 (w) \alpha_0 (r+w-z) |(V\alpha_0)(z)| |\psi_s (X)|^{6} \right]^{2/5}  \\ & \times 
 \Big[  \int dX dr dz dw \int_0^1 d\kappa \, \alpha_0 (r) \alpha_0 (w) \alpha_0 (r+w-z) |(V\alpha_0)(z)|  \\ & \hspace{7cm} \times \left| \nabla \psi_s (X+\kappa h (r+z)/2)  \right|^2 \Big]^{3/5}
 \\ \leq \; & C h^{6/5} \| \psi_s \|_{H^1}^{6/5} \end{split} \]
for a constant $C$ depending on $\| \alpha_0\|_1$, $\| \alpha_0 \|_\infty$, $\| r^6 \alpha_0 \|_1$, and $\| V\alpha_0 \|_1$ (in the second integral, we shifted the variable $X \to X-\kappa h (r+z)/2$). Similarly, we bound
\[ \begin{split}  |G^{(2)}_{2,s} (X)| \leq \; & h |\psi_s (X)| \int dr dz dw \int_0^1 d\kappa \,  |w+z| \, \alpha_0 (r) \alpha_0 (w) \alpha_0 (r+w-z) |(V\alpha_0) (z)| \\ &\hspace{4cm} \times  |\psi_s (X + h (r+z)/2)|  \, \left|\nabla \psi_s (X + \kappa h (w+z)/2) \right| \end{split} \]
which gives
\[ \begin{split} \| G^{(2)}_{2,s} \|_{6/5}^{6/5} \leq \; & h^{6/5} \left( \int dr dz dw \int_0^1 d\kappa \, |w+z|^6 \alpha_0 (r) \alpha_0 (w) \alpha_0 (r+w-z) |(V\alpha_0) (z)| \right)^{1/5} \\ &\times \int dX dr dz dw \int_0^1 d\kappa \, \alpha_0 (r) \alpha_0 (w) \alpha_0 (r+w-z) |(V\alpha_0) (z)| \\ &\hspace{1cm} \times  |\psi_s (X)|^{6/5} |\psi_s (X + h (r+z)/2)|^{6/5} \left|\nabla \psi_s (X + \kappa h (w+z)/2) \right|^{6/5}  \\ \leq \; & C h^{6/5} \left[ \int dX dr dz dw \int_0^1 d\kappa 
\alpha_0 (r) \alpha_0 (w) \alpha_0 (r+w-z) |(V\alpha_0) (z)|  |\psi_s (X)|^{6}  \right]^{1/5} 
\\ &\times \left[ \int dX dr dz dw \int_0^1 d\kappa 
\alpha_0 (r) \alpha_0 (w) \alpha_0 (r+w-z) |(V\alpha_0) (z)|  |\psi_s (X+h (r+z)/2)|^{6}  \right]^{1/5} 
\\ &\times \left[ \int dX dr dz dw \int_0^1 d\kappa 
\alpha_0 (r) \alpha_0 (w) \alpha_0 (r+w-z) |(V\alpha_0) (z)| \right. \\ &\left. \hspace{7cm} \times  \left|\nabla \psi_s (X + \kappa h (w+z)/2) \right|^{2} \right]^{3/5}  
\\ \leq \; &C h^{6/5} \| \psi_s \|^{6/5}_{H^1} \end{split} \]
Finally, we bound the term $G^{(2)}_{3,s}$. To this end we proceed analogously, noticing that
\[  \begin{split} |G^{(2)}_{3,s} (X)| \leq  \; &h \int dr dz dw \int_0^1 d\kappa \,  |z-r| \, \alpha_0 (r) \alpha_0 (w) \alpha_0 (r+w-z) |(V\alpha_0) (z)| \\ &\times |\psi_s (X+h(z+w)/2)| \, |\psi_s (X + h (r+z)/2)|  \, \left|\nabla \psi_s (X + \kappa h (z-r)/2) \right| \end{split} \]
and therefore concluding that
\[ \begin{split} \| G^{(2)}_{3,s} & \|_{6/5}^{6/5} \\ \leq \; &h^{6/5} \left( \int dr dz dw \int_0^1 d\kappa \, |z-r|^6 \alpha_0 (r) \alpha_0 (w) \alpha_0 (r+w-z) |(V\alpha_0) (z)| \right)^{1/5}  \\ & \times \int dX dr dz dw \int_0^1 d\kappa 
\alpha_0 (r) \alpha_0 (w) \alpha_0 (r+w-z) |(V\alpha_0) (z)| \\ & \hspace{1cm} \times  |\psi_s (X+h (z+w)/2)|^{6/5}  \, |\psi_s (X + h (r+z)/2)|^{6/5} \left|\nabla \psi_s (X + \kappa h (z-r)/2) \right|^{6/5} 
\\ \leq \; & C h^{6/5}  \left[ \int dX dr dz dw \int_0^1 d\kappa 
\alpha_0 (r) \alpha_0 (w) \alpha_0 (r+w-z) |(V\alpha_0) (z)|  |\psi_s (X+h (z+w)/2)|^{6}  \right]^{1/5} 
\\ &\times \left[ \int dX dr dz dw \int_0^1 d\kappa 
\alpha_0 (r) \alpha_0 (w) \alpha_0 (r+w-z) |(V\alpha_0) (z)|  |\psi_s (X+h (r+z)/2)|^{6}  \right]^{1/5} 
\\ &\times \left[ \int dX dr dz dw \int_0^1 d\kappa 
\alpha_0 (r) \alpha_0 (w) \alpha_0 (r+w-z) |(V\alpha_0) (z)| \left|\nabla \psi_s (X + \kappa h (z-r)/2) \right|^{2} \right]^{3/5}  
\\ \leq \; &C h^{6/5} \| \psi_s \|^{6/5}_{H^1} \end{split} \]
 \end{proof}
 
\begin{lemma} \label{lm:Gj}
Consider, from (\ref{eq:error-1}), the error terms 
\[ \begin{split} 
G^{(1)}_{1,s} (X) = \; &-\frac{1}{h^7} \int dr dz dw \, \alpha_0 (r/h) \overline{\psi}_s (X+(r+z)/2) \alpha_0 (w/h)  \wt{\xi}_s (X + (z+w)/2,  r + w - z)  \\ &\hspace{7cm} \times (V\alpha_0) (z/h) \psi_s ( X+(z-r)/2)\\
G^{(1)}_{2,s} (X) = \; &-\frac{1}{h^7} \int dr dz dw \, \alpha_0 (r/h)  \overline{\wt{\xi}_s (X+(r+w)/2 , w)} \alpha_0 ((r+w-z)/h) \psi_s (X+(z+w)/2)
\\ & \hspace{7cm} \times  (V\alpha_0) (z/h) \psi_s ( X+(z-r)/2)   \\ 
G^{(1)}_{3,s} (X) = \; &-\frac{1}{h^5} \int dr dz dw \, \alpha_0 (r/h)  \overline{\wt{\xi}_s (X+(r+w)/2 , w)}  \wt{\xi}_s (X + (z+w)/2,  r + w - z) \\ & \hspace{7cm} \times  (V\alpha_0) (z/h) \psi_s ( X+(z-r)/2)  \\
G^{(1)}_{4,s} (X) = \; &- \frac{1}{h^5} \int dr dz \alpha_0 (r/h) (V\alpha_0) (z/h) (\gamma_s - \bar{\alpha}_s \alpha_s) (X+r/2, X+z-r/2)  \psi_s (X+(z-r)/2) \\
G^{(1)}_{5,s} (X) = \; &- \frac{1}{h^3} \int dr dz \, \alpha_0 (r/h)  \gamma_s (X+r/2, X+z -r/2) V(z/h) \wt{\xi}_s (X+(z-r)/2,z) 
\end{split} \]
Then $\| G^{(1)}_{j,s} \|_{6/5} \leq C h$ for $j =1,2$ and $\| G_{j,s}^{(1)} \|_{6/5} \leq C h^{1/2}$ for $j =3,4,5$, for a constant $C$ independent of $h$ and of the time $s$. 
\end{lemma}

\begin{proof}
We start with the error term $G^{(1)}_{1,s}$. We have, using H\"older inequality with the dual exponents $6$ and $6/5$, 
\[ \begin{split} 
\| G_{1,s}^{(1)} \|_{6/5}^{6/5} = \; & \frac{1}{h^{42/5}} \int dX \left| \int dr dz dw \, \alpha_0 (r/h) \alpha_0 (w/h) (V\alpha_0) (z/h) \right. \\ & \hspace{1cm}  \left. \times \overline{\psi}_s (X+(r+z)/2) \,  \psi (X+(z-r)/2) \wt{\xi}_s (X+ (z+w)/2, r+w -z) \right|^{6/5}  \\ \leq \; & \frac{1}{h^{42/5}} \left[ \int dr dw dz \, \alpha_0 (r/h) \alpha_0 (w/h) (V\alpha_0) (z/h)  \right]^{1/5} \\ & \times \int dX dr dz dw \, \alpha_0 (r/h) \alpha_0 (w/h) (V\alpha_0) (z/h)  \\ & \hspace{1cm}  \times  |\psi_s (X+(r+z)/2)|^{6/5} |\psi (X+(z-r)/2)|^{6/5} |\wt{\xi}_s (X+ (z+w)/2, r+w -z)|^{6/5} \\
\leq \; & \frac{1}{h^{42/5}} \left[ \int dr dw dz \, \alpha_0 (r/h) \alpha_0 (w/h) (V\alpha_0) (z/h)  \right]^{1/5} \\ & \times \left[ \int dX dr dz dw \, \alpha_0 (r/h) \alpha_0 (w/h) (V\alpha_0) (z/h) |\psi_s (X+(r+z)/2)|^{6} \right]^{1/5} 
\\ &\times \left[ \int dX dr dz dw \, \alpha_0 (r/h) \alpha_0 (w/h) (V\alpha_0) (z/h) |\psi_s (X+(z-r)/2)|^{6} \right]^{1/5}
\\ & \times \left[ \int dX dr dz dw \, \alpha_0 (r/h) \alpha_0 (w/h) (V\alpha_0) (z/h) |\wt{\xi}_s (X+(z+w)/2, r+w-z)|^{2} \right]^{3/5} \\
\leq \; & h^{3/5}  \| \alpha_0 \|_\infty^{3/5}  \| \alpha_0 \|_1^{9/5} \| V\alpha_0 \|_1^{6/5}  \| \psi_s \|_6^{12/5}  \| \wt{\xi} \|_2^{6/5} \leq C h^{6/5}
\end{split} \] 
 where we used the bound $\| \wt{\xi}_s \|^2_{2} \leq Ch$ from Proposition \ref{prop}. The term $G^{(1)}_{2,s}$ can be bounded analogously. As for the term $G_{3,s}^{(1)}$, we write $\xi_s (x,y) = \wt{\xi}_s ((x+y)/2, x-y)$, and we notice that  
 \[ \begin{split}
 G_{3,s}^{(1)} (X)  =\; & \frac{1}{h^5} \int dr dz dw \, \alpha_0 (r/h)  (V\alpha_0) (z/h)  \\ & \hspace{1cm} \times \psi_s ( X+(z-r)/2)  \overline{\xi_s (X+ w + r /2 ,X + r/2)}  \xi_s (X + w + r/2, X +z -r/2) \\ 
=\; & \frac{1}{h^5} \int dr dz \, \alpha_0 (r/h)  (V\alpha_0) (z/h)  \psi_s ( X+(z-r)/2) (\overline{\xi}_s \xi_s) (X+r/2, X+z-r/2) 
\end{split} \]
where $(\overline{\xi}_s \xi_s ) (x,y)$ denotes the kernel of the operator $\overline{\xi}_s \xi_s$, i.e. \[ (\overline{\xi}_s \xi_s) (x,y) = \int dz \overline{\xi}_s (x,z) \xi_s (z,y) = \int dz \overline{\xi_s (x,z)} \xi_s (z,y)\] (here we use the fact that, by definition, $\xi_s$ is symmetric, like $\alpha_s$). Therefore
\[ \begin{split} \| G_{3,s}^{(1)} & \|_{6/5}^{6/5} \\ \leq \; & \frac{1}{h^6} \int dX \, \left| \int dr dz \, \alpha_0 (r/h)  (V\alpha_0) (z/h)  \psi_s ( X+(z-r)/2) (\overline{\xi}_s \xi_s) (X+r/2, X -z-r/2)  \right|^{6/5}  \\
\leq \; & \frac{1}{h^6} \left[ \int dr dz \,  \alpha_0 (r/h)  (V\alpha_0) (z/h) \right]^{1/5} \\ 
&\times  \int dX dr dz \, \alpha_0 (r/h)  (V\alpha_0) (z/h)  |\psi_s ( X+(z-r)/2)|^{6/5} |(\overline{\xi}_s \xi_s)  (X+r/2, X+z-r/2)|^{6/5} \\
\leq \; & \frac{\| \alpha_0 \|^{1/5}_1 \| V\alpha_0 \|^{1/5}_1}{h^{24/5}}  
\left[\int dX dr dz \, \alpha_0 (r/h)  (V\alpha_0) (z/h)  |\psi_s ( X+(z-r)/2)|^{3} \right]^{2/5}  \\ & \hspace{2cm} \times \left[ \int dX dr dz \, \alpha_0 (r/h)  (V\alpha_0) (z/h) |(\overline{\xi}_s \xi_s)  (X+r/2, X -z-r/2)|^{2}\right]^{3/5} \\ 
\leq \; & \frac{\| \alpha_0 \|^{3/5}_1 \| \alpha_0 \|_\infty^{3/5} \| V\alpha_0 \|^{6/5}_1 \| \psi_s \|_3^{6/5}  (\tr \, (\overline{\xi}_s \xi_s) ^2)^{3/5}}{h^{3/5}}
\leq \; C h^{3/5} \end{split} \]
where we used that \[ \begin{split} \tr \, (\overline{\xi}_s \xi_s)^2  & = \int dx_1 dx_2 dx_3 dx_4 \, \overline{\xi}_s (x_1, x_2) \xi_s (x_2 , x_3) \overline{\xi}_s (x_3, x_4) \xi_s (x_4, x_1) \\ & \leq  \int dx_1 dx_2 dx_3 dx_4 \, |\xi_s (x_1, x_2)|^2  |\xi_s (x_3, x_4)|^2  \\ & = \| \xi_s \|_2^4 \leq C h^2 \end{split} \]
from Proposition \ref{prop}. The term $G_{4,s}^{(1)}$ can be controlled using the bound $\tr \, (\gamma_s - \bar{\alpha}_s \alpha_s) \leq C h$ from Proposition \ref{prop}. Since $(\gamma_s - \bar{\alpha}_s \alpha_s) \geq 0$, this implies immediately that $\tr \, (\gamma_s - \bar{\alpha}_s \alpha_s)^2 \leq C h^2$. Therefore  
\[\begin{split} 
\| G_{4,s}^{(1)} \|_{6/5}^{6/5} \leq \; & \frac{1}{h^6} \int dX \Big| \int dr dz \, \alpha_0 (r/h) (V\alpha_0) (z/h) (\gamma_s - \bar{\alpha}_s \alpha_s) (X+r/2, X+z-r/2) \\ 
&\hspace{10cm} \times \psi_s (X+(z-r)/2) \Big|^{6/5} \\ 
\leq \; & \frac{1}{h^6} \left[ \int dr dz \, \alpha_0 (r/h) (V\alpha_0) (z/h) \right]^{1/5}\\ & \hspace{1cm} \times \int dX dr dz \, \alpha_0 (r/h) (V\alpha_0) (z/h) |(\gamma_s - \bar{\alpha}_s \alpha_s) (X+r/2, X+z-r/2)|^{6/5}  \\ 
&\hspace{10cm} \times |\psi_s (X+(z-r)/2)|^{6/5} \\
\leq \; &\frac{\| \alpha_0 \|_1^{1/5} \| V\alpha_0 \|_1^{1/5}}{h^{24/5}} \left[ \int dX dr dz \, \alpha_0 (r/h) (V\alpha_0) (z/h) \,  |\psi_s (X+(z-r)/2)|^{3}  \right]^{2/5} \\ & \hspace{1cm} \times  \left[ \int dX dr dz \, \alpha_0 (r/h) (V\alpha_0) (z/h) \, | (\gamma_s - \bar{\alpha}_s \alpha_s) (X+r/2, X+z-r/2)|^2 \right]^{3/5}   \\ \leq \; &\frac{\| \alpha_0 \|_1^{3/5} \| \alpha_0 \|_\infty^{3/5} \| V\alpha_0 \|_1^{6/5} \| \psi \|_3^{6/5} (\tr \, (\gamma_s - \bar{\alpha}_s \alpha_s)^2)^{3/5}}{h^{3/5}} \leq \; C h^{3/5} \end{split} \] 
Finally, we compute
\[ \begin{split} \| G_{5,s}^{(1)} \|_{6/5}^{6/5} \leq \; & \frac{1}{h^{18/5}} \int dX \left| \int dr dz \, \alpha_0 (r/h) V(z/h) \gamma_s (X+r/2, X+z-r/2) \wt{\xi}_s (X+(z-r)/2, z) \right|^{6/5} \\ 
\leq \; & \frac{1}{h^{18/5}} \left[ \int dr dz \, \alpha_0 (r/h) V(z/h)  \right]^{1/5} \\ &\times  \int dX dr dz \, \alpha_0 (r/h) V(z/h) |\gamma_s  (X+r/2, X+z-r/2)|^{6/5}  |\wt{\xi}_s (X+(z-r)/2, z)|^{6/5} \\
\leq \; & \frac{\| \alpha_0 \|^{1/5}_1 \| V \|^{1/5}_1}{h^{12/5}} \left[ \int dX dr dz \, \alpha_0 (r/h) V(z/h) |\gamma_s  (X+r/2, X+z-r/2)|^2 \right]^{3/5}\\ &\hspace{1cm} \times \left[  \int dX dr dz \, \alpha_0 (r/h) V(z/h)  |\wt{\xi}_s (X+(z-r)/2, z)|^3 \right]^{2/5} \\ 
\leq \; & h^{3/5}  \| \alpha_0 \|^{3/5}_1 \| V \|^{2/5}_\infty \| \alpha_0 \|_\infty^{3/5} \| V \|^{4/5}_1 (\tr \gamma_s^2)^{3/5} \| \wt{\xi}_s \|_3^{6/5} \leq \; C h^{3/5} 
\end{split} \]
using that, by Proposition \ref{prop}, $\tr \, \gamma_s^2 \leq h$ and $\| \wt{\xi}_s \|_3 \leq \| \wt{\xi}_s \|_{H^1} \leq h^{-1/2}$. 
\end{proof}
 
 \appendix
 
\section{Strichartz Estimates}
\label{sec:strich}

Let $n \geq 3$. A pair $(r,q) \in [1,\infty) \times [1,\infty)$ is called admissible if $2 \leq r \leq 2n/(n-2)$ and $(2/q) = n (1/2 - 1/r)$. Let
\[ \phi_f (t) = \int_0^t ds \, U(t-s) f(s) \]
For any two admissible pairs $(r,q)$ and $(\rho,\delta)$ we have
\[ \| \phi_f \|_{L^q ([0,T] , L^r (\bR^n))} \leq C \| f \|_{L^{\delta'} ([0,T], L^{\rho'} (\bR^3))} \]
where $(\delta', \rho')$ are dual indices to $(\delta,\rho)$ (i.e. $1/\delta + 1/ \delta' = 1/\rho + 1/\rho' = 1$). In our analysis we use Strichartz estimates with $(r,q) = (2,\infty)$ and $(\rho,\delta) = (6,2)$. In the latter case, the Strichartz estimate is referred to as the endpoint Strichartz estimate, and was proven in \cite{KT}.

\end{document}